\RequirePackage{fix-cm}
\documentclass[twocolumn]{svjour3}          
\smartqed  

\usepackage{graphicx}
\usepackage{amssymb}
\usepackage{amsmath}
\usepackage{xcolor}
\usepackage{graphicx}
\usepackage{placeins}

\newcommand{\bit}{\begin{itemize}}
\newcommand{\eit}{\end{itemize}}
\newcommand{\bd}{\begin{description}}
\newcommand{\ed}{\end{description}}
\newcommand{\ben}{\begin{enumerate}}
\newcommand{\een}{\end{enumerate}}
\newcommand{\bqn}{\begin{equation*}\begin{aligned}}
\newcommand{\eqn}{\end{aligned}\end{equation*}}
\newcommand{\bqnn}{\begin{equation}\begin{aligned}}
\newcommand{\eqnn}{\end{aligned}\end{equation}}
\newcommand{\bt}{\begin{thm}}
\newcommand{\et}{\end{thm}}
\newcommand{\bl}{\begin{lem}}
\newcommand{\el}{\end{lem}}
\newcommand{\bp}{\begin{prop}}
\newcommand{\ep}{\end{prop}}
\newcommand{\bc}{\begin{cor}}
\newcommand{\ec}{\end{cor}}
\newcommand{\bdefn}{\begin{defn}}
\newcommand{\edefn}{\end{defn}}
\newcommand{\brem}{\begin{rem}}
\newcommand{\erem}{\end{rem}}
\newcommand{\bproof}{\begin{proof}}
\newcommand{\eproof}{\end{proof}}
\newcommand{\bex}{\begin{ex}}
\newcommand{\eex}{\end{ex}}
\newcommand{\R}{\mathbb{R}}

\newcommand{\bcs}{\begin{cases}}
\newcommand{\ecs}{\end{cases}}

\newcommand{\expp}{\text{e}}

\newcommand{\se}{\text{SE}(2)}

\newcommand{\INT}{\int\limits}
\newcommand{\SUM}{\sum\limits}

\newcommand{\abs}[1]{\lvert {#1} \rvert}

\begin{document}

\title{A geometric model of multi-scale orientation preference maps via Gabor functions
}
%

\author{E. Baspinar        \and
        G. Citti		 \and 
        A. Sarti     
}


\institute{Universit\`{a} di Bologna-Dipartimento di Matematica and CAMS/CNRS-EHESS, Paris} 
           

\date{Received: date / Accepted: date}

\maketitle

\section{Introduction}
\graphicspath{{figures/JMIV_June_2017/}}
As it is well known simple cells in the primary visual cortex are organized in structures called orientation preference maps. This special organization has been studied with geometric instruments starting by the work of Petitot and Tondut \cite{petitot1999vers}. In that study orientation maps  were obtained as the superposition of randomly weighted orientation fields corresponding to all possible orientation angles around the pinwheels (see the geometric explanations of Petitot \cite{petitot2003neurogeometry} related to the pinwheels). A different model, always based only on orientation was introduced by Barbieri et al. \cite{barbieri2011coherent}, \cite{barbieri2012uncertainty}. In that paper the orientation preference structure was recovered starting from the observation that its Fourier transform is concentrated on an annulus. This model as the previous recalled one, is based on properties apparently independent of the other aspects of the cortical models.  Additionally to those studies, the models, in terms of its cortical orientation and orientation-frequency selectivity, which were provided by Bressloff and Cowan \cite{bressloff2003functional}, \cite{bressloff2001geometric} and the model proposed for the cortical spatio-temporal selective behavior by Barbieri et. al. \cite{barbieri2014cortical} could be useful references for the reader.

In this article we present a new model for the generation  of orientation preference maps, considering both orientation and scale features. Hypercolumns of the simple cell receptive profiles are the fundamental units of the set of receptive profiles and they build a 2-dimensional sub-group of rotation-dilation at each point $(x,y)$ of the retinal plane $M\subset \R^2$. In other words the base variables are the spatial components $(x,y)\in M$ and the intrinsic variables are orientation and scale parameters $(\theta, \sigma)\in [0, \pi )\times \R^+$. Having two intrinsic variables in hand, we can either fix scale and obtain the orientation map of simple cells or we can employ a range of scale values and obtain a multiscale orientation map. In this way the model integrates several visual features observed in neurophysiology, psychophysics and neuroimaging experiments and provides a more precise orientation map.

The main novelty of our approach is that the orientation map description is strongly related to the functionality of the cortex, and simple cell responses in presence of a visual stimulus. Indeed we start with a random stimulus $I$ on the retinal plane, and obtain the responses of the cells through a linear filtering with translated, rotated and dilated Gabor functions. Finally, we employ integration of the output over fiber and maximum selection in order to select the prevalent orientation and scale. This whole procedure starting with obtaining the simple cell responses and ending with application of integration over fiber and maximum selection over the full set of receptive profiles is called lifting. Consequently we propose to obtain orientation maps by employing a lifting of noise stimulus through the functional structure of the cortex. We will outline that this corresponds to a Bargmann transform \cite{bargmann1967hilbert} in the reducible representation of the $\se$ group, which is different than the case in \cite{barbieri2011coherent} where Barbieri et. al. considered the irreducible representation. Hence our model is neural based. 

The theoretical criterion underpinning the modeling we propose in this paper relies on the so-called neurogeometrical approach described by Citti and Sarti \cite{citti2006cortical}, Petitot and Tondut \cite{petitot1999vers}, Sarti et. al. \cite{sarti2008symplectic}. Following this approach, processing capabilities of sensorial cortices and particularly of the visual cortex are modeled based on the geometrical structure of neural connectivity.
Global and local symmetries of the visual stimuli are inherited by the cortical structure that presents their invariances (see Sanguinetti et. al. \cite{sanguinetti2010model}). Then the structure is defined on group of invariances that are also spaces, meaning Lie groups. Particularly the simple cells are sensitive to local position and orientation features of stimuli, which are elements of the roto-translation group $\se$. The corresponding Lie algebra and its integral curves model neural connectivity between cells. Moreover, since the algebra is not commutative, it is possible to pose an uncertainty principle, whose minimization gives rise to the shape of receptive profiles of the simple cells. The model has been extended to other variables such as scale by Sarti et. al. \cite{sarti2008symplectic}, and to other cell types such as complex cells sensitive to movement by Barbieri et. al. \cite{barbieri2014cortical} and Cocci et. al. \cite{cocci2015cortical}. In \cite{citti2014neuromathematics} and \cite{sarti2015constitution},  a neurogeometrical field theory has been introduced by Sarti and Citti to model connectivity between different cortices and it has been shown that harmonic analysis on the neurogeometry excited by the stimulus accounts for the constitution of perceptual units, while in \cite{sarti2015individuation} semiotic forms have been obtained through the same principle by Sarti and Piotrowski.

Orientation maps of V1 have been introduced in \cite{barbieri2011coherent} by Barbieri et. al. as Bargmann transform in the irreducible representation of $\se$, while in the present article here a model of orientation maps is proposed in terms of a reducible representation, that is more neurophysiologically plausible. Then all the principal morphologies present in the visual cortex are modeled in a compact way in the neurogeometrical framework.

To our knowledge none of the other approaches (such as orientation map construction methods proposed by Barbieri et.al. \cite{barbieri2012uncertainty}, Petitot \cite{petitot2003neurogeometry}, multi-scale approach of Linderberg \cite{lindeberg1998feature} and other methods proposed based on differential geometry by Franken et.al. \cite{franken2007nonlinear}, Ben-Shahar and Zucker \cite{ben2004geometrical}) is able to cover such a variety of forms and visual phenomena starting from the very first principles.

As a general consideration about the choice of the receptive profile model, let us recall that receptive field models consisting of cascades of linear filters and static non-linearities may be adequate to account for responses to simple stimuli such as gratings and random checkerboards, but their predictions of responses to complex stimuli such as natural scenes are only approximately correct. A variety of mechanisms such as response normalization, gain controls, cross-orientation suppression, intra-cortical modulation can intervene to change radically the shape of the profile. Then any static and linear model for receptive profiles has to be considered just as a very first approximation of the complex behavior of a real dynamic receptive profile, which is not perfectly described by any of the static wavelet frames.

For example derivatives or difference of Gaussian functions are very good approximations of the behavior of classical receptive profiles of the simple cells. Let us outline that such families of Gaussian functions account just for symmetric receptive profiles, while shift in phase is not considered. This could induce to think that generalized Gabor filtering is more flexible. On the other hand it is true that the majority of receptive profiles of simple cells in the primary visual cortex are even and odd symmetric, and it is an open issue to evaluate the importance to discard a minority of asymmetric profiles.

In the specific model which we propose in this article, we have used only the Gabor filters without any shift in phase. In this case the Gabors can be easily replaced with derivatives of Gaussians, without loss of generality. However the choice we made based on Gabors allows to extend the model to the true distribution of profiles in the primary visual cortex (including asymmetric receptive profiles with shifts in phase), i.e., to a neurophysiologically coherent generic model of the visual cortex, which is not possible with derivatives of Gaussian functions. The reader can find more information about some models employing alternative choices of receptive profiles in terms of Gaussian derivatives in the works of Koenderink \cite{koenderink1984structure}, \cite{koenderink1987representation} and Lindeberg \cite{lindeberg2013computational}.

We test the model at different scales, in order to represent properties of orientation maps in different cortical areas where the scale of the receptive profile changes. Our simulation results are compared with neural experimental results. A comparison will be provided with a previous model based on the Bargmann transform in the irreducible representation of the $\se$ group, outlining that the new model is more physiologically motivated.
Moreover we remark that it is possible to extend the model in order that additional visual features such as frequency and phase are taken into account.

In Section \ref{sec:sec2} we explain receptive profiles of simple cells and describe the group structure with its geometrical properties. Then we give explicitly the procedure of the construction of cortical map in Section \ref{sec:sec_Procedure}. Afterwards, in Section \ref{sec:bargmannTransform}, we show that Gabor functions are minimizers of an uncertainty principle and the filtering with the Gabors can be interpreted as a Bargmann transform in reducible representations. Then we provide the relation of Bargmann transform to the orientation map construction procedure and we compare it to another method using the Bargmann transform with Gabor functions in irreducible representations on the Fourier domain. Finally, in Section \ref{sec:secExperiments} we present the simulation results of the model and compare them to experimental results given in the literature.

\section{Receptive profiles of simple cells}
\label{sec:sec2}

\subsection{Receptive fields and receptive profiles}

The simple cells of visual areas evoke impulse responses to stimulus applied on the retinal plane $M\subset \R^2$. Every simple cell is identified by its \emph{receptive field} (RF) which is defined as the domain of the retina to which the cell is sensitive and connected through the retino-geniculo-cortical paths. Once a RF is stimulated it evokes a spike response. 

In classical sense a RF contains on and off regions, i.e., positive and negative contrast regions, respectively. The decomposition of RF into those regions depends on the nature of the cell response given to light and dark luminance Dirac stimulations. The response is realized by the simple cell receptive profile. Receptive profile (RP) of a simple cell is defined on RF and it is simply the impulse response of the cell. Conceptually it is the measurement of the response of the corresponding cell to a stimulus at some point $(x,y)\in M$. We denote the RP at the retinal position $(x,y)\in M$ with orientation $\theta\in [0,\pi)$ and scale $\sigma\in\R^+$ by $\Psi_{(x,y,\theta,\sigma)}:M\times [0, \pi)\times\R^+\rightarrow \mathbb{C}$. The simple cells of the primary visual cortex are strongly oriented and they are sensitive to several visual features, in particular to orientation and scale. Their RPs are often interpreted as Gabor functions \cite{gabor1946theory} since Gabor functions are mathematically convenient for encoding such features as Daugman \cite{daugman1985uncertainty} explained based on a minimum uncertainty condition. In the literature other models employing alternative choices of RPs in terms of Gaussian derivatives were proposed as well, following the works of Koenderink \cite{koenderink1984structure}, \cite{koenderink1987representation} where he pointed out the resemblance between simple cell receptive profiles and Gaussian derivative kernels. The reader can refer to Lindeberg \cite{lindeberg2013computational} where he proposed a family of functions in terms of Gaussian derivatives as a natural choice of the simple cell receptive profile with respect to certain symmetry properties.


\subsection{The set of receptive profiles}
\label{sec:The_set_of_receptive}

Once the retinal layer is activated by some visual stimulus $I(x_0,y_0)\in \R$, at the point $(x_0,y_0)\in M$ the simple cells process the retinal stimulus through their RPs which are denoted by $\Psi_{(x_0,y_0)}$ where the sub-index refers to the corresponding spatial position on $M$ at which $\Psi$ is centered. Each RP at the point $(x_0,y_0)$ is dependent on a preferred orientation $\theta$ and a scale $\sigma\in\R^+$ (see Figure \ref{fig:Gabor_Rotated_Real_2} and Figure \ref{fig:Gabors_with_different_scales}). The set of RPs is obtained through translation to the point $(x_0,y_0)$ and rotation by $\theta$, i.e.,
\begin{align}\label{eq:coordTransform}
\begin{split}
T_{(x_0,y_0,\theta,\sigma)}(\xi,\eta)= &
\begin{pmatrix}
x_0 \\ y_0
\end{pmatrix}+\expp^{\sigma}\begin{pmatrix}
\cos(\theta) & \;-\sin(\theta)\\
\sin(\theta) & \;\cos(\theta)
\end{pmatrix}
\begin{pmatrix}
\xi \\ \eta
\end{pmatrix}
\\ = & (x,y),
\end{split}
\end{align}
applied on the Gabor mother function
\begin{equation}\label{eq:motherGaborPrimitive}
\Psi_0(\xi,\eta)=\expp^{-(\xi^2+\eta^2)}\expp^{i2\eta}.
\end{equation}

General expression of Gabor functions obtained from the mother function is given by
\begin{equation}\label{eq:theoreticalPartGabor}
\Psi_{(x_0,y_0,\theta,\sigma)}(x,y)=\Psi_0(T^{-1}_{(x_0,y_0,\theta,\sigma)}(x,y)).
\end{equation}

\graphicspath{{figures/JMIV_June_2017/}}
\begin{figure}
\centerline{\includegraphics[scale=0.2,trim={0cm 0cm 0cm 0cm},clip]{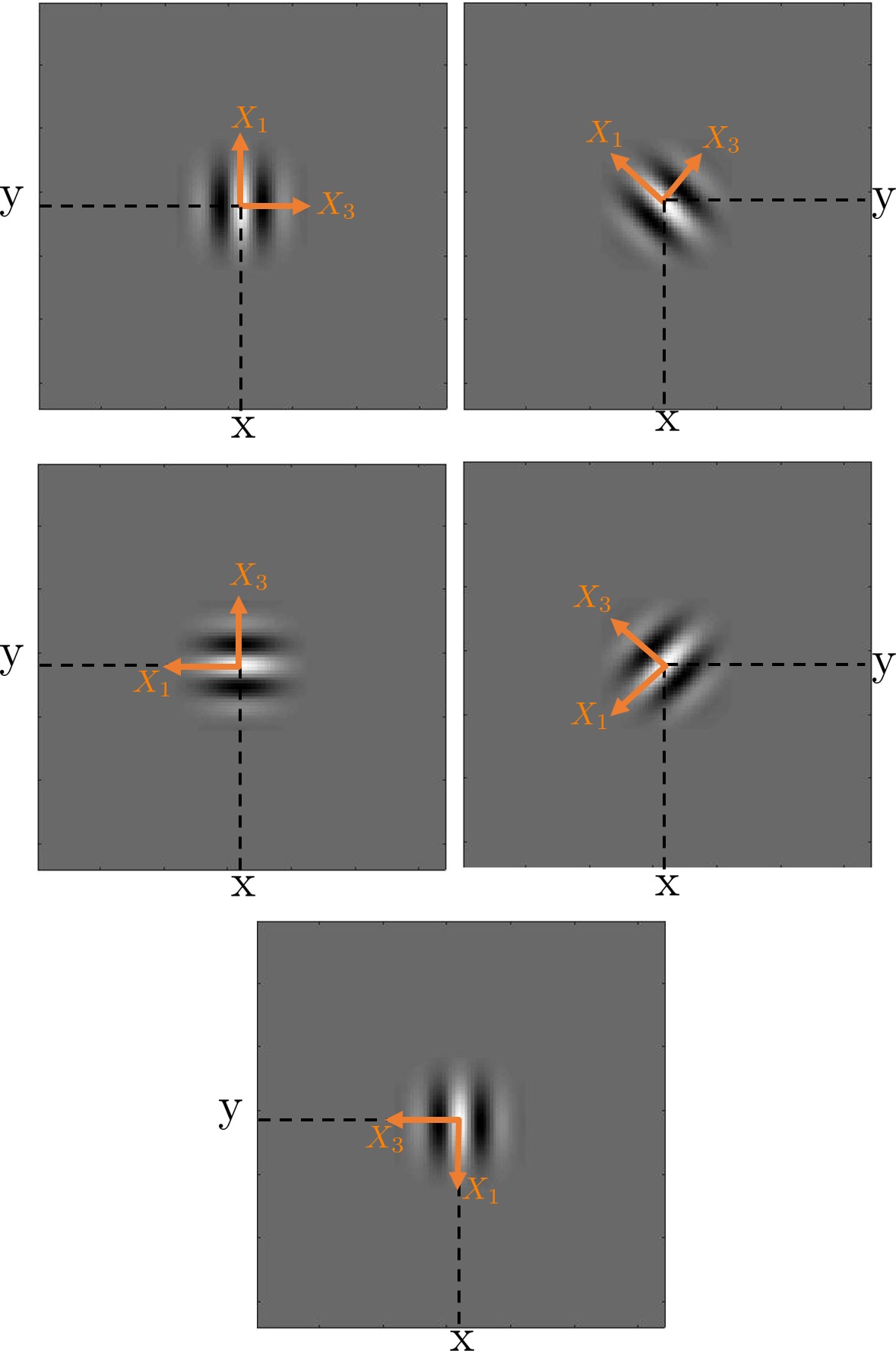}}
\caption{Real (even) part of rotated Gabor filter $\Psi_{(x_0,y_0,\theta,\sigma)}$ centered at $(x_0,y_0)$, with scale $\sigma=8$ and $\theta=0$ (top left), $\theta=\frac{\pi}{4}$ (top right), $\theta=\frac{\pi}{2}$ (middle left), $\theta=\frac{3\pi}{4}$ (middle right), $\theta=\pi$ (bottom). The direction $X_3$ is the image gradient direction while $X_1$ is the tangent direction. }
\label{fig:Gabor_Rotated_Real_2}
\end{figure}

\begin{figure}
\centerline{\includegraphics[scale=0.2,trim={0cm 0cm 0cm 0cm},clip]{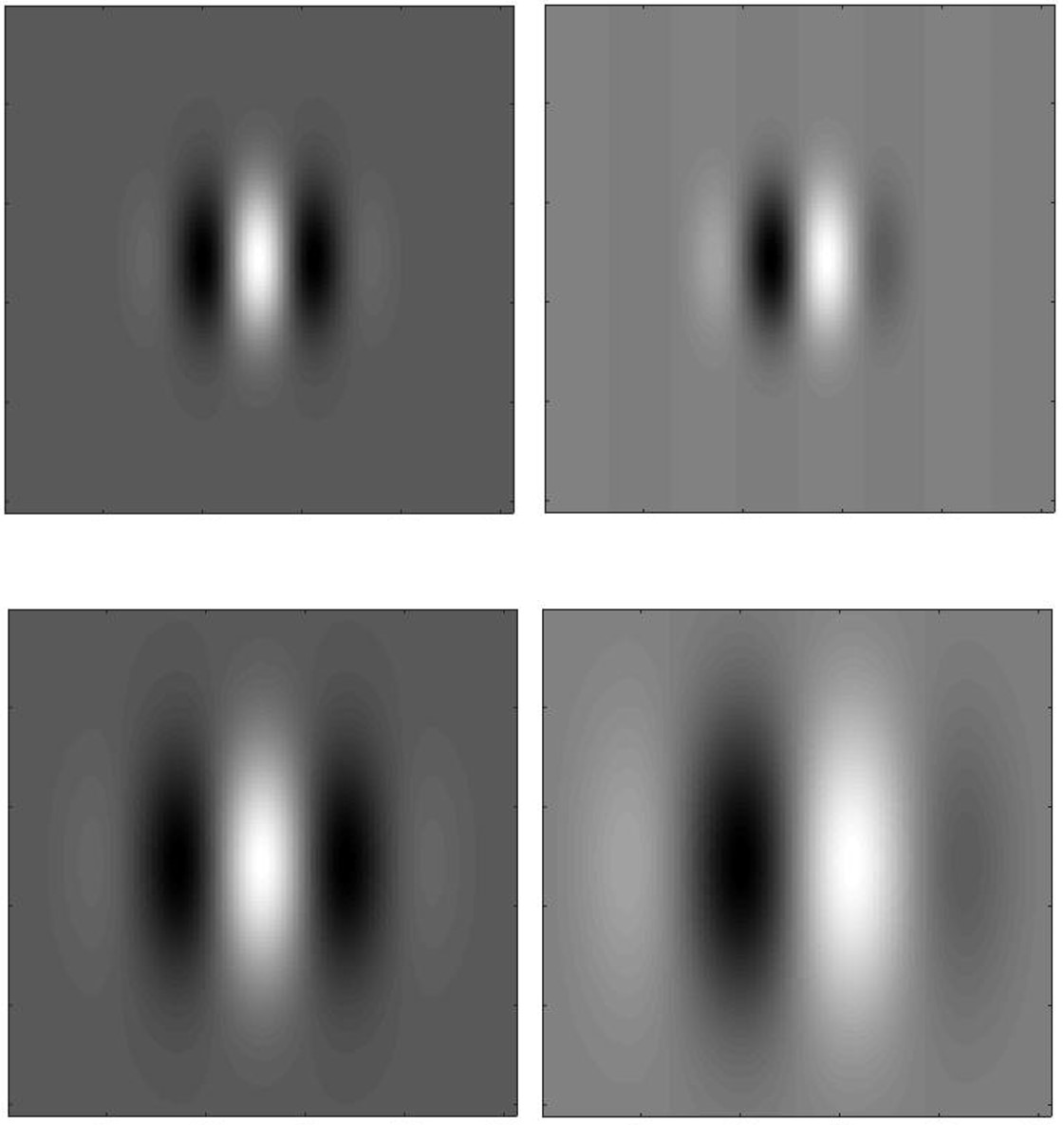}}
\caption{Real (or even, left column) and imaginary (or odd, right column) parts of the Gabor filters with different scales (top and bottom rows).}
\label{fig:Gabors_with_different_scales}
\end{figure}

Note that we find the transformation law of the group
\begin{equation}
G\simeq \{T_{(x_0,y_0,\theta_0,\sigma_0)}: (x_0,y_0,\theta,\sigma)\in \R^2\times [0,\pi)\times \R^+ \},
\end{equation}
by applying the coordinate transform given by \eqref{eq:coordTransform} successively as follows:
\begin{align}
\begin{split}
\, & T_{(x_1, y_1, \theta_1,\sigma_1)} T_{(x_0, y_0,\theta, \sigma)}(\xi,\eta)\\
= & \begin{pmatrix}
x_1 \\ y_1
\end{pmatrix}+\expp^{\sigma_1} R_{\theta_1}\begin{pmatrix}
x_0 \\ y_0
\end{pmatrix}+\expp^{\sigma_1+\sigma}R_{\theta_1+\theta}\begin{pmatrix}
\xi \\ \eta
\end{pmatrix}\\
= & T_{(x_2, y_2, \theta_2, \sigma_2)}(\xi,\eta),
\end{split}
\end{align}
where $R_{\theta}$ represents the rotation matrix and
\begin{align}
\begin{split}
\sigma_2=\sigma_1+\sigma,\quad \theta_2=\theta_1+\theta,\\
\begin{pmatrix}
x_2 \\ y_2 
\end{pmatrix}=\begin{pmatrix}
x_1 \\ y_1
\end{pmatrix}+\expp^{\sigma_1}R_{\theta_1}\begin{pmatrix}
x_0 \\ y_0
\end{pmatrix}.
\end{split}
\end{align}

We refer to the explanations provided by Sarti et. al. in \cite{sarti2008symplectic} for more details.

\subsection{Functional connectivity of the cortex}

The hypercolumns are endowed with internal isotropic short range connections which we specifically call \emph{vertical connections}. The vertical connections do not provide the inter-hypercolumnar interactions and without such inter-connections, the hypercolumns located at different retinal points $(x_0,y_0)\in M$ would be isolated from each other. We know from the neurophysiological results (see the works of Bosking et. al. \cite{bosking1997orientation}, Das and Gilbert \cite{das1995long}) that there are long ranged, strongly anisotropic connections between hypercolumns. This second type of connections within the primary visual cortex is called \emph{horizontal connections}. They play the main role in inter-columnar information flow, i.e., contour integration and image inpainting.

Moreover Bosking et al. \cite{bosking1997orientation} observed that the horizontal connections link preferentially the simple cells at different spatial locations $(x_0,y_0)\in M$ but corresponding to the same orientation (approximately). In other words, the horizontal connections characterize the contour integration along the aligned curve fragments with approximately same orientations, respecting the \emph{saliency} (see the saliency description of Wertheimer \cite{wertheimer1938laws}) of the global structure obtained through the integration. Contour integration in a salient way is closely related to the existence of specific connectivities within the primary visual cortex, which are named as \emph{association fields} by Field et al. \cite{field1993contour} confirming the anisotropic behavior of the horizontal connections in the psychophysical level.

In order to implement this functional connectivity we associate to each receptive profile $\Psi_{(x_0,y_0,\theta,\sigma)}$ the following 1-form
\begin{equation}
\omega_{(\theta,\sigma)}=\expp^{-\sigma}(-\sin(\theta)dx+\cos(\theta)dy),
\end{equation}
where $dx,\,dy\in T^{\ast}M$ represent the covector fields dual to the vector fields $\partial_x,\partial_y\in TM$. The 1-form $\omega$ is the main instrument describing the orientation selectivity of a simple cell since it selects the direction along the vector field
\begin{equation}
X_3=\expp^{\sigma}(-\sin(\theta)\partial_x+\cos(\theta)\partial_y),
\end{equation}
and the vector $X_3\vert_{(x_0,y_0,\theta,\sigma)}$ at point $(x_0,y_0)$ gives the image gradient at that point corresponding to the receptive profile $\Psi_{(x_0,y_0,\theta,\sigma)}$. The direction along $X_3$ is associated with the orientation angle which the simple cells at $(x_0, y_0)$ are sensitive to (see also Figure \ref{fig:Gabor_Rotated_Real_2}). Furthermore, with the additional exponential $\expp^{-\sigma}$, the 1-form $\omega$ weights the contour fragment at $(x_0, y_0)$, lying orthogonal to $X_3$, in such a way that the fragment corresponding to the same scale as $\omega$ produces the highest simple cell response magnitude. In short $\omega_{(\theta, \sigma)}$ is the main instrument which renders both orientation and scale selectivity of the primary visual cortex simple cells. 

Finally we find the horizontal left invariant vector fields as
\begin{equation}
\operatorname{ker}\omega=\operatorname{span}\{X_1, X_2, X_4 \},
\end{equation}
where
\begin{align}\label{eq:horizontalVectorFields}
\begin{split}
X_1= & \expp^{\sigma}(\cos(\theta)\partial_x+\sin(\theta)\partial_y),\\
X_2= & \partial_{\theta},\\
X_4= & \partial_{\sigma}.
\end{split}
\end{align}
Here we note that due to the fact that
\begin{align}
\begin{split}
[X_1, X_2]= & -X_3,\\
[X_1, X_4]= & -X_1,
\end{split}
\end{align}
the horizontal vector fields are non-commutative. Yet they span the whole tangent bundle together with their commutators, i.e.,
\begin{equation}
TM=\operatorname{span}\{X_1, X_2, X_4, [X_1,X_2]\}.
\end{equation}
That is, the horizontal vector fields given by \eqref{eq:horizontalVectorFields} fulfill the H\"{o}rmander condition \cite{hormander1967hypoelliptic}.



\section{The model of multi-scale orientation maps}
\label{sec:sec_Procedure}

In this section we present our model of orientation cortical maps. 
As we explained in the introduction, we propose that cortical maps 
are obtained via a two step procedure: First the simple cells act on a random stimulus, 
and consequently maximally activated orientation and scale are selected, producing the cortical map.

\label{sec:The_lifting_mechanism}
 The response given to a stimulus by a simple cell with orientation preference $\theta$, scale $\sigma$ and located at $(x_0, y_0)\in M$ is expressed by
\begin{align}\label{eq:filteringExpressionFin}
\begin{split}
O_{(\theta,\sigma)}(x_0,y_0)=&\INT_M I(x,y)\Psi_{(x_0,y_0,\theta,\sigma)}(x,y)\,dx\,dy.
\end{split}
\end{align}
See Figure \ref{fig:receptive_profile_full_response} for a visualization of such outputs. For every retinal point $(x_0,y_0)$ a particular value of orientation is selected via integration on the fiber:

\graphicspath{{figures/JMIV_June_2017/}}
\begin{figure}
\centerline{\includegraphics[scale=0.4,trim={0cm 0cm 0cm 0cm},clip]{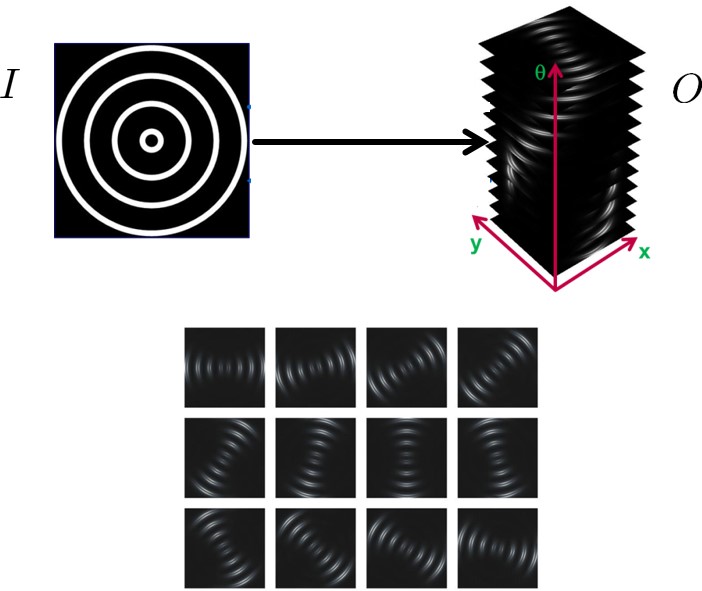}}
\caption{Lifting with a fixed scale is applied to an image $I$ and the full set of simple set responses (outputs) $O$ is obtained. The responses corresponding to each horizontal layer associated to an orientation angle is at the bottom.}
\label{fig:receptive_profile_full_response}
\end{figure}

\begin{equation}\label{eq:ourThetaCondition}
\overline{\theta}(x_0, y_0)= \frac{1}{2}\operatorname{arg}\Big(\INT_0^{\pi}\operatorname{Re}\Big\{O_{(\theta,\sigma)}(x_0,y_0)\Big\}\expp^{i\theta}d\theta\Big).
\end{equation}
We considered here  just the real part of the output but alternative choices are possible, like for example  the energy or the imaginary part of the output. We refer to \cite{sarti2009functional} for more information about such choices.

Lindeberg, in \cite{lindeberg1998feature}, describes a scale selection technique in terms of Gaussian derivatives normalized by scale. Basically the method finds extrema over scales corresponding to normalized receptive field responses by scale. A similar approach in our particular framework associated to Gabor functions is considered and scale selectivity is provided by the maximum of the output at the point $(x_0,y_0)$ over the scale fiber at the selected value of $\overline\theta$:
\begin{equation}\label{eq:ourSigmaCondition}
\overline{\sigma}(x_0, y_0)= \underset{{\sigma\in\R^+}}{\operatorname{argmax}}\Big(\operatorname{Re}\Big \{O_{(\overline{\theta}, \sigma)}(x_0,y_0)\Big\}\Big).
\end{equation}
Let us note that we employ maximum selectivity \eqref{eq:ourSigmaCondition},  for selecting the scale value, as Sarti et. al. did in \cite{sarti2008symplectic}, while we use the integration over fiber \eqref{eq:ourThetaCondition} in order to find the orientation preference over the fiber at the point $(x_0,y_0)\in M$. 
This procedure allows us to achieve a more robust orientation selectivity. Here we assume that generically there is a unique maximum, so that it is equivalent the order in which we select $\overline{\theta}$ and $\overline{\sigma}$. Note that the procedure  described by \eqref{eq:ourThetaCondition} and \eqref{eq:ourSigmaCondition} is done for  every fixed point $(x_0, y_0)$ on the retinal plane and the selected orientations $\overline{\theta}(x_0, y_0)$ and scales $\overline{\sigma}(x_0,y_0)$ are represented at the corresponding fiber locations $(x_0, y_0)\in M$.
In such a way we obtain the multi-scale orientation map $\overline{\theta}(x,y)$ that is represented in Figure \ref{fig:multiScale_FullProcedure2} with the same type of color map as in the classical case Figure \ref{fig:BoskingFigure}. The overall procedure for obtaining cortical maps is schematized in Figure \ref{fig:orientation_preference_map_diagram}. 

\graphicspath{{figures/JMIV_June_2017/}}
\begin{figure*}
\centerline{\includegraphics[scale=0.3,trim={0cm 0cm 0cm 0cm},clip]{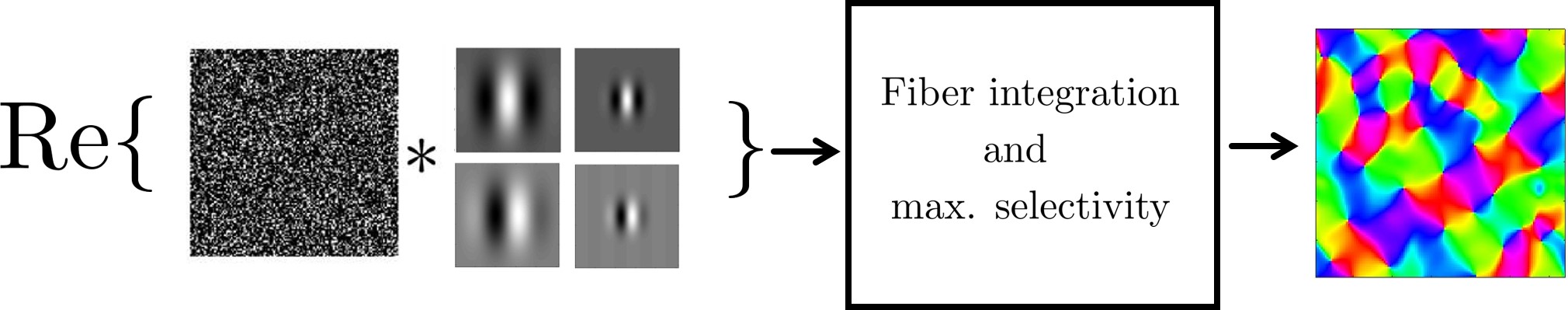}}
\caption{Image with white noise is filtered by Gabor filters with different scales. Integration over fiber and maximum selectivity among the integrated fiber values are applied on the real part of the filtering result. The orientation preference map is obtained by assigning a certain color to each orientation value.}
\label{fig:orientation_preference_map_diagram}
\end{figure*}
 

\begin{figure}
\centerline{\includegraphics[scale=0.5,trim={0cm 0cm 0cm 0},clip]{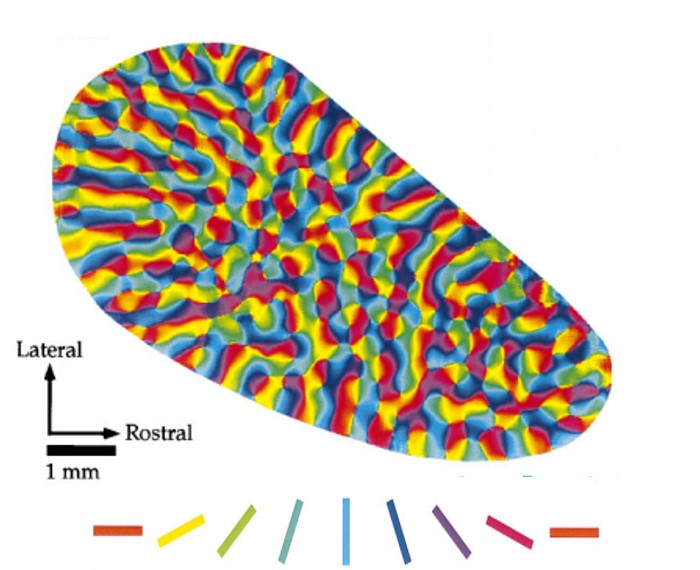}}
\caption{Orientation preference map taken from \protect\cite[Figure 1]{bosking1997orientation}. It was obtained via vector summation of data recorded for each angle by using optical imaging.}
\label{fig:BoskingFigure}
\end{figure}

This procedure corresponds to the lifting of a general stimulus $I(x,y) $ provided by simple cells circuitry. We explicitly note here that cortical orientation maps will be obtained by using the lifting of a random stimulus. This choice is motivated also by the fact that experimentally cortical maps arise in the early post natal period in absence of any visual experience just in presence of an intrinsic random basal activation (see the studies of Jegelka et. al. \cite{jegelka2006prenatal}, Tanaka et. al. \cite{tanaka2004roles}, Bednar and Miikkulainen \cite{bednar2004constructing}). A refinement of the orientation maps is performed subsequently by activation patterns based on random waves (see the results provided by Cang et. al. \cite{cang2005development}, Stellwagen and Shatz \cite{stellwagen2002instructive}).

Note that convolution with a Gabor filter will provide a smooth function. 
Indeed the Gabor is simply a Gaussian function multiplied by a complex exponential. 
The resulting function will then be a smooth function depending on the variance of the Gaussian, 
which is the scale. Finally the orientation selection will provide smooth functions, 
with values in $S^1$. It is well known that even harmonic function with values in $S^1$ develop 
vortices, which will be a model for the pinwheel.

It is natural to build feature cortical maps by means of Gabor functions, 
since they are strictly related to all the functional geometry. 
In fact we will see that they arise as minimizers of the uncertainty principle in this setting.

\section{Orientation maps as cortical Bargmann transforms}\label{sec:Gabor functions as minimizing coherent states of the uncertainty principle}
\label{sec:bargmannTransform}
\subsection{An uncertainty principle}

Orientation maps have been constructed by Barbieri et. al. in \cite{barbieri2012uncertainty} where 
an uncertainty principle related to the functional geometry of the cortex and its  non-commutative structure were used. 

%

The uncertainty principle in its general form always applies in presence of two self-adjoint non-commutating vector fields $P_1$ and $P_2$. 
In our framework, as given by Folland in \cite{folland2016harmonic}, it is written as the following:
\begin{proposition}\label{prop:uncertaintyStatement}
Let us denote $\mathcal{H}$ an Hilbert space endowed with the scalar product $\langle .\,,\, . \rangle$. Consider two self-adjoint vector fields $P_1$ and $P_2$ on $\mathcal{H}$. Then the following inequality holds:
\begin{equation}\label{eq:uncertaintyPrinciple}
\abs{\langle f, [P_1, P_2] f\rangle}\leq 2 \|P_1 f \| \|P_2 f\|,
\end{equation}
for all $L^2(\R^2)$ functions $f$ in the domain of $[P_1, P_2]$. 
\end{proposition}

\begin{proof}
Since $P_1$ and $P_2$ are self-adjoint, we can write that
\begin{align}
\begin{split}
\langle f,\, [P_1,\, P_2]f \rangle= & \langle f,\, (P_1P_2-P_2P_1)f   \rangle\\
= & \langle P_1 f,\,P_2f \rangle- \langle P_2 f,\,P_1 f \rangle\\
= & 2i\operatorname{Im}\{\langle P_1 f, P_2 f \rangle  \}.
\end{split}
\end{align}
We employ the Cauchy-Schwarz inequality and write:
\begin{align}\label{CauchySch}
\begin{split}
\langle f,\, [P_1,\, P_2]f \rangle\leq 2\abs{ \langle P_1 f, P_2 f  \rangle} \leq 2 \|P_1 f \| \| P_2f \|.
\end{split}
\end{align}
\qed
\end{proof}

The first inequality in \eqref{CauchySch} becomes an equality when $\langle P_1 f, P_2 f \rangle$ is purely imaginary and the second one turns into equality when $P_1f$ and $P_2f$ are linearly dependent. As a consequence 
minima of uncertainty inequality (i.e., for that the inequality \eqref{eq:uncertaintyPrinciple} turns into equality) satisfy the following equation
\begin{equation}\label{eq:minimizerCondition}
P_1 f= i\lambda P_2 f .
\end{equation}
The condition $P_1 f= i\lambda P_2$ with $\lambda\in \R$ gives the minimizers which are called coherent states. 

\subsection{Gabors as minimizers of the uncertainty principle}\label{sec:Gabor functions as minimizing coherent states of the uncertainty principle}

We have seen that Gabor functions are defined on the 2-dimensional retinal plane, generated by the action on a mother filter of 
$T^{-1}_{(x_0,y_0,\theta,\sigma)}$ where $T_{(x_0,y_0,\theta,\sigma)}$ is defined in \eqref{eq:coordTransform}. Accordingly the differential of $T^{-1}$ sends the vector fields $X_1$ and $X_2$ acting in the 4-dimensional manifold of variables $(x,y,\theta, \sigma)$ to new vector fields $Y_1$ and $Y_2$ defined on the retinal plane  as:
\begin{align}
\begin{split}\label{eq:VFsY1_Y2}
(dT^{-1})(X_1)= & \partial_{\xi},\\ 
(dT^{-1})(X_2)= &\eta\partial_{\xi}-\xi\partial_{\eta}.
\end{split}
\end{align}  
It is well-known that Gabor filters are minimal of the uncertainty principle 
in the Heisenberg group, less known is that they are indeed also coherent states 
related to the non-commutating vector fields $Y_1$ and $Y_2$ induced by the functional architecture. As a result they  satisfy the analogous of \eqref{eq:minimizerCondition}. 
Precisely Gabor functions of the type $\Psi_{(x_0,y_0,\theta,\sigma)}$ given by \eqref{eq:theoreticalPartGabor} satisfy
\begin{align}
\begin{split}
Y_1\Psi_{(x_0,y_0,\theta,\sigma)}(x,y)=-2\xi \Psi_0(\xi,\eta),\\
Y_2\Psi_{(x_0,y_0,\theta,\sigma)}(x,y)=-i2\xi\Psi_0(\xi,\eta),
\end{split}
\end{align}
which fulfill \eqref{eq:minimizerCondition} for $\lambda=-1$.

\subsection{Interpretation of cortical maps as Bargmann transform of a random stimulus}

Let us recall that the operator associated to coherent states is 
called Bargmann transform: 

\begin{align}\label{eq:Bargmann}
\begin{split}
 (B^{SE(2)}I)(x,y,\theta, \sigma):=&
\langle \Psi_{(x_0,y_0,\theta, \sigma)},\; I\rangle_{L^2(\R^2)}\\
= & \INT_{\R^2}\Psi_{(x_0,y_0,\theta, \sigma)}(\xi, \eta)
I(\xi, \eta)d\xi, d\eta.
\end{split}
\end{align}

In particular, the response of simple cells, 
being defined by \eqref{eq:filteringExpressionFin}  as the convolution with a Gabor coherent state, 
can be interpreted as a Bargmann transform in the functional architecture: 
\begin{equation}\label{eq:output_bargmann} 
O_{(\theta, \sigma)}(x,y ) = (B^{SE(2)}I)(x,y,\theta, \sigma). 
\end{equation}
As a consequence, orientation maps are 
associated to Bargmann transform of a random stimulus.

\section{Comparison with previous models of cortical maps}

\subsection{Superposition of random waves}

One of the first models for construction of orientation preference maps is proposed by Petitot in \cite{petitot2003neurogeometry} where the map is obtained through the superposition of randomly weighted complex sinusoids
\begin{equation}
\SUM_{k=1}^{k=N}c_k\expp^{i2\pi\big(x\cos(2\pi k/N)+y\sin(2\pi k/N)\big)},
\end{equation}
with $N$ denoting the number of frequency samples and where the coefficients $c_k\in [0, 1]$ are the white noise. 

In this way the functional role of Gabor functions as receptive profiles is disregarded since the orientation map was constructed via direct superimposition of the waves with randomly generated magnitudes, avoiding that Gabors naturally process the stimulus by lifting it to the phase space of corresponding intrinsic variables. More specifically in this procedure the complex sinusoid functions are not localized while it is known from neurophysiological experiments that the orientation selectivity is performed locally by the simple cells (see for example the work of Field and Tolhurst \cite{field1986structure}).

\subsection{Bargmann transform of irreducible representations}

In the model proposed by Barbieri et. al. in \cite{barbieri2012uncertainty}, the orientation map is 
built starting from coherent states in the irreducible representation.
\begin{definition}
The representation of a group $G$ is a map $\Phi: G\rightarrow A(V)$, from the group $G$ to the space of automorphisms of a vector space $V$, such that $\Phi$ is compatible with the group law. The representation will be denoted by $(\Phi, V)$, and it is called \emph{irreducible} if it has no proper group subrepresentation $(\Phi, W)$, where $W$ is a subspace of $V$.
\end{definition}
 
Taking the Fourier transform of the vector fields $Y_1$ and $Y_2$ defined in 
\eqref{eq:VFsY1_Y2}, we obtain: 
\begin{equation}
\mathcal{F}(Y_1f)=iz_1\hat{f},\quad \mathcal{F}(Y_2f)=(z_2\partial_{z_1}-z_1\partial_{z_2})\hat{f}.
\end{equation}

We can write those vector fields also in terms of polar coordinates $(z_1,z_2)=(\Omega\cos(\varphi), \Omega\sin(\varphi))$ with $\Omega\in\R^+$ and $\varphi\in S^1$. In this case the fields become
\begin{equation}\label{eq:vectorFieldsInFourierSpace}
\hat{Y}_1\hat{f}=i\Omega\cos(\varphi)\hat{f},\quad \hat{Y}_2\hat{f}=\partial_{\varphi}\hat{f}.
\end{equation}

The vector fields $\hat{Y}_1$ and $\hat{Y}_2$ do not contain any radial derivative and only depend on the angular direction in the Fourier space.  Therefore they act independently on every circle, of arbitrary radius $\Omega$. Then it is possible to restrict the action of these vector fields to any circle with radius $\Omega$ on the Fourier space separately (see the explanations of Sugiura \cite{sugiura1990unitary} for details). 
This is the reason why the vector fields $Y_1, Y_2$ on the whole space (in the Fourier domain as well) are called reducible, 
while  $\hat{Y}_1$ and $\hat{Y}_2$ which cannot be further reduced once $\Omega$ is fixed, are called 
irreducible.

If we write the coherent state condition \eqref{eq:minimizerCondition} on the Fourier domain in terms of $\hat{Y}_1$ and $\hat{Y}_2$, 
\begin{equation}
\hat{Y}_1\hat{f}=i\lambda 	\hat Y_2\hat{f},
\end{equation}
we find the coherent states
\begin{align}
\begin{split}
\hat{\Psi}^{\Omega}_{(x_0,y_0,\theta,\sigma)}(\varphi)=\hat{\Psi}_{(x_0,y_0,\theta,\sigma)}(\Omega\cos(\varphi),\Omega\sin(\varphi)),\\
\end{split}
\end{align}
where $\hat{\Psi}_{(x_0,y_0,\theta,\sigma)}$ is the Fourier transform of the 
Gabor filters, while $\hat{\Psi}^{\Omega}_{(x_0,y_0,\theta,\sigma)}$ is a function of the angular variable, defined on the circle of radius $\Omega$.

In \cite{barbieri2011coherent} and \cite{barbieri2012uncertainty} Barbieri et. al. use the family of coherent states obtained for fixed value of $\sigma$, and for a single value of $\Omega$
\begin{equation}
\quad\quad\quad\quad\quad\quad\quad\quad\hat{\Psi}^{\Omega}_{(x_0,y_0,\theta)}.
\end{equation}

In perfect analogy with equation \eqref{eq:filteringExpressionFin} the Bargmann transform in these variables is expressed as the operator with kernel $\hat{\Psi}^{\Omega}_{(x_0,y_0,\theta)}$ as:
\begin{align}\label{eq:barbieriMethodPW_Construction_Init}
\begin{split}
B^{\Omega} g(x,y,\theta):=&
\langle \hat{\Psi}^{\Omega}_{(x_0,y_0,\theta)},\; g\rangle_{L^2(S^1)}\\
= & \INT_0^{2\pi}\hat{\Psi}^{\Omega}_{(x_0,y_0,\theta)}(\varphi)g(\varphi)d\varphi.
\end{split}
\end{align}

In \cite{barbieri2011coherent} and \cite{barbieri2012uncertainty}, this transform is applied to a white noise  $g$ defined on the annulus (on Fourier domain). 
For every point $(x,y)$ an orientation is selected by means of an integration analogous to the one expressed in \eqref{eq:ourThetaCondition}:

\begin{equation}\label{eq:orientationinfourier}
\overline{\theta}^{\Omega}(x,y)= \frac{1}{2}\operatorname{arg}\Big(\INT_0^{\pi}\Big\{ B^{\Omega} g(x,y,\theta)\Big\}\expp^{i\theta}d\theta\Big).
\end{equation}

In this way they find an orientation preference at point $(x,y)$ 
which depends on the fixed value of $\Omega$ and they obtain orientation preference maps (with no scale parameters).

Although both our model and the model proposed in \cite{barbieri2011coherent} and \cite{barbieri2012uncertainty} by Barbieri et. al. make use of the idea of Bargmann transform they differ on three points. 

Firstly our method employs coherent states corresponding to the reducible representations while the other one uses the states restricted to the irreducible representations in the Fourier domain. 

Secondly we start from a noise generated on the real domain and apply a Bargmann transform, while the other method introduces the noise in the Fourier domain on the irreducible representations, and apply the Bargmann transform in the Fourier space. The choice made in the present paper here is physiologically more plausible since experimentally cortical maps can arise in the early post natal period in absence of any external stimulus just in presence of a random basal activation (see Bednar and Miikkulainen \cite{bednar2004constructing} and Jegelka et. al. \cite{jegelka2006prenatal}). 
The present model has the potential to provide a reasoning and an explanation of how the formation of cortical maps occurs at the neurophysiological level.

The third main difference is that the present model can also consider scale selectivity while in the other model the scale is fixed. More generally it is possible to extend the present model in order to include other visual features by using higher dimensional Gabor functions.

\section{Experiments}
\label{sec:secExperiments}
\label{sec:Experimental framework}\label{sec:Experimental_framework}

We consider a stimulus $I(x,y)$ of  $128\times 128$ pixels with random values generated from a uniform distribution over $[-1, 1]$ at each pixel.

We obtain the total set of simple cell responses via the linear filtering of the test image with rotated and translated Gabor filter bank as described in \eqref{eq:filteringExpressionFin} with different scale values $\sigma$. Then we represent the selected orientation $\overline{\theta}(x,y)$ and $\overline{\sigma}(x,y)$, via \eqref{eq:ourThetaCondition} and \eqref{eq:ourSigmaCondition}, at every point $(x,y)$ on the $128\times 128$ image plane.

Previously in the literature it was reported from the physiological experiments of Bosking \cite{bosking1997orientation} (see Figure \ref{fig:BoskingFigure}) that the orientation preference map had certain characteristics (see the explanations of Bressloff and Cowan \cite{bressloff2003spherical}, and Petitot \cite[p.27]{citti2014neuromathematics}, \cite[p.87]{petitot2003neurogeometry}). To begin with, orientation preferences on the map are distributed almost continuously across the cortex and the pinwheel architecture is crystalline-like. In other words there is a regular lattice of pinwheels on the orientation preference map with a certain spatial periodicity. Furthermore the orientation map contains three types of points as described by Petitot \cite[p.87]{petitot2003neurogeometry}, namely : a) Regular points around which the orientation iso-lines are parallel (the zones with regular points are called \emph{linear zones}), b) Singular points which are located at the center of the pinwheels (Those singularities might have positive or negative \emph{chirality}. That is, when we turn around a pinwheel in the clockwise direction, the orientations turn in the clockwise direction - positive chirality - or in the counter-clockwise direction - negative chirality. The pinwheels represent opposite chiralities when they are adjacent to each other), c) Saddle points at the center of regions where iso-orientation lines bifurcate (the case where two iso-orientation lines start from the same pinwheel and arrive at opposite pinwheels).

We will see that in the present study, we are able to produce all the three kinds of points. In the first experiment we consider different fixed scales and apply \eqref{eq:filteringExpressionFin} and \eqref{eq:ourThetaCondition} to obtain orientation maps: Results are shown in Figure \ref{fig:pinwheels_With_VFs}, where orientation maps are visualized and the three kinds of points are outlined. 

\begin{figure}
\centerline{\includegraphics[scale=0.3,trim={0cm 0cm 0cm 0cm},clip]{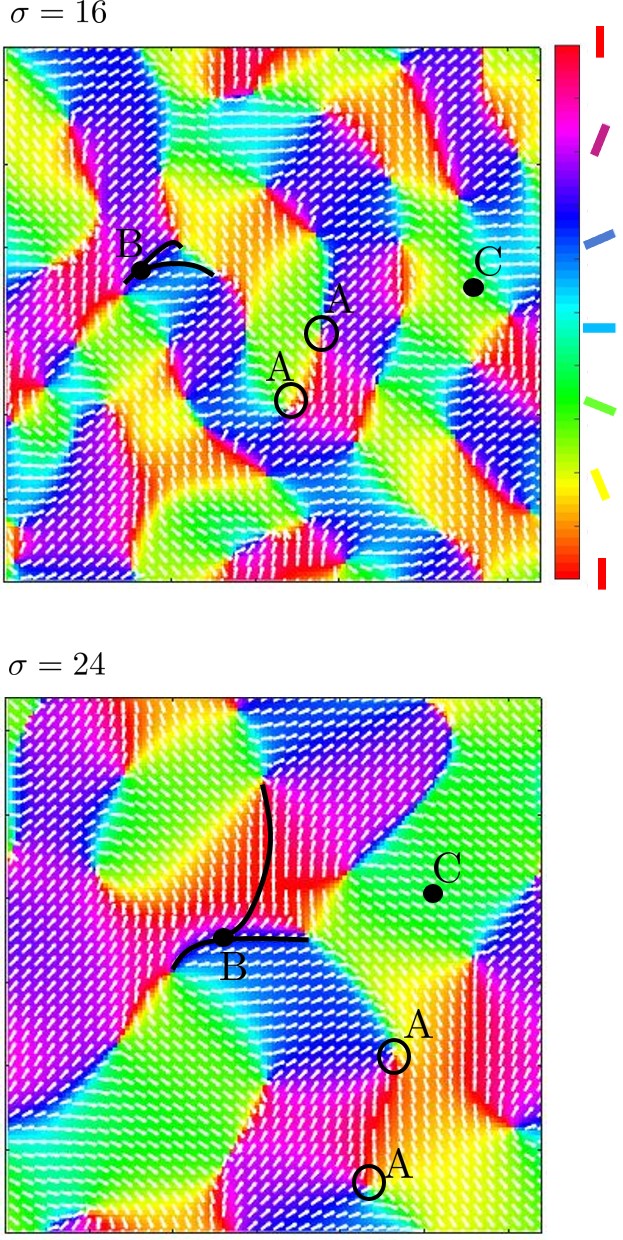}}
\caption{Orientation preference maps obtained through our model with scale $\sigma=16$ (top) and $\sigma=24$ (bottom), adjacent pinwheels with opposite chiarilities (points A), saddle points (points B) and linear zones (points C) represented by a single color. White lines represent the orientation correspondence at each point.}
\label{fig:pinwheels_With_VFs}
\end{figure}

Figure \ref{fig:periodicity_Computation_Figure} shows the cross-correlation between simulated cortical maps  where several picks are present to testify the crystalline structure of the map. Notice that the periodicity of the peaks is linearly dependent on the scale of Gabor filters employed for the construction of the map.

\begin{figure}
\centerline{\includegraphics[scale=0.25,trim={0cm 0cm 0cm 0cm},clip]{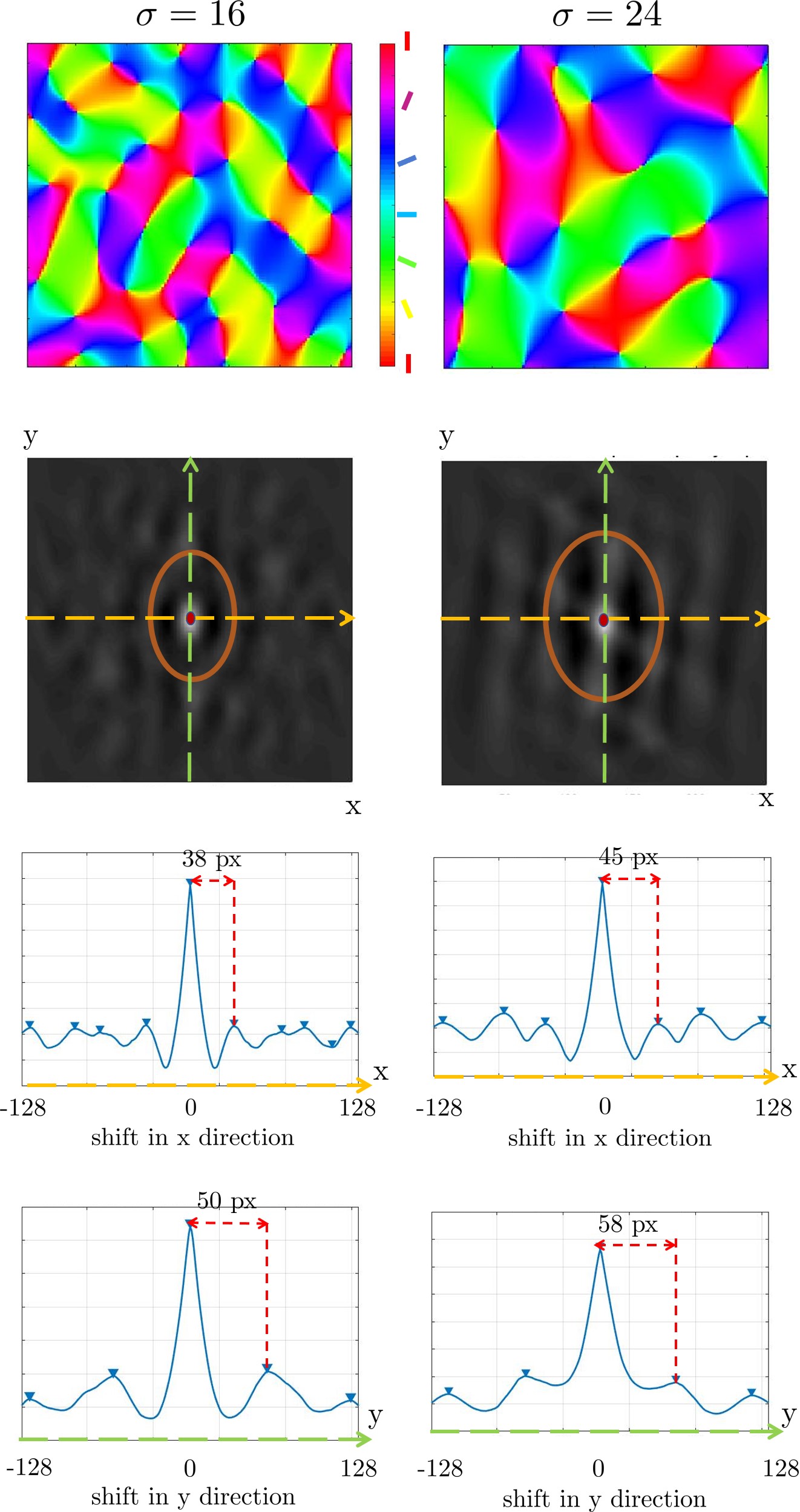}}
\caption{Results obtained by Gabors of scale values $\sigma=16$ (left column) and $\sigma=24$ (right column). Top: Orientation preference maps, Second row: Orientation preference map cross-correlations. The average of the vertical and horizontal axes of ellipses (orange) representing the second peak values around the first peak corresponding to the exact match due to no shift indicates the spatially periodic configuration of pinwheel grid structure of the orientation preference maps. Third row: Cross-correlation values with respect to the shifts in x direction along the profile line (orange dashed arrow in the second row). Bottom: Cross-correlation values with respect to the shifts in y direction along the profile line (green dashed arrow in the second row). Finally the spatial shift corresponding to the second peaks for $\sigma=16$ is found as $44$ pixels while for $\sigma=24$ it is $52$ pixels approximately.}
\label{fig:periodicity_Computation_Figure}
\end{figure}

The size of the pinwheel structure is also strictly correlated to the scale of the Gabor filters, as shown in Figure \ref{fig:multiscale_pinwheels2_Paper}.

\begin{figure}
\centerline{\includegraphics[scale=0.2,trim={0cm 0cm 0cm 0cm},clip]{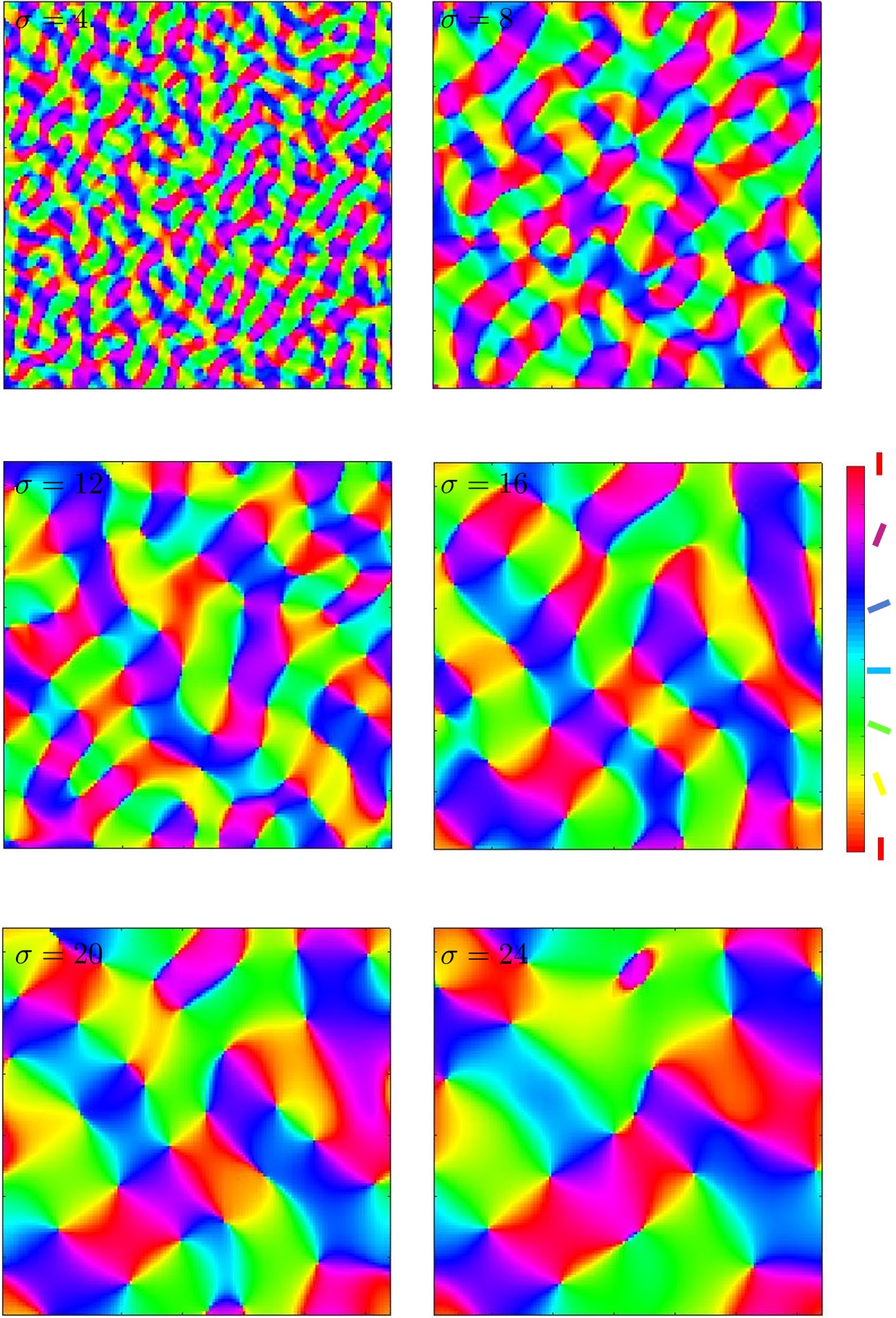}}
\caption{Orientation maps obtained with Gabors of scales $\sigma=4$ (top left), $\sigma=8$ (top right), $\sigma=12$ (middle left), $\sigma=16$ (middle right), $\sigma=20$ (bottom left), $\sigma=24$ (bottom right) in pixels.}
\label{fig:multiscale_pinwheels2_Paper}
\end{figure}

\begin{figure}
\centerline{\includegraphics[scale=0.25,trim={0cm 0 0 0},clip]{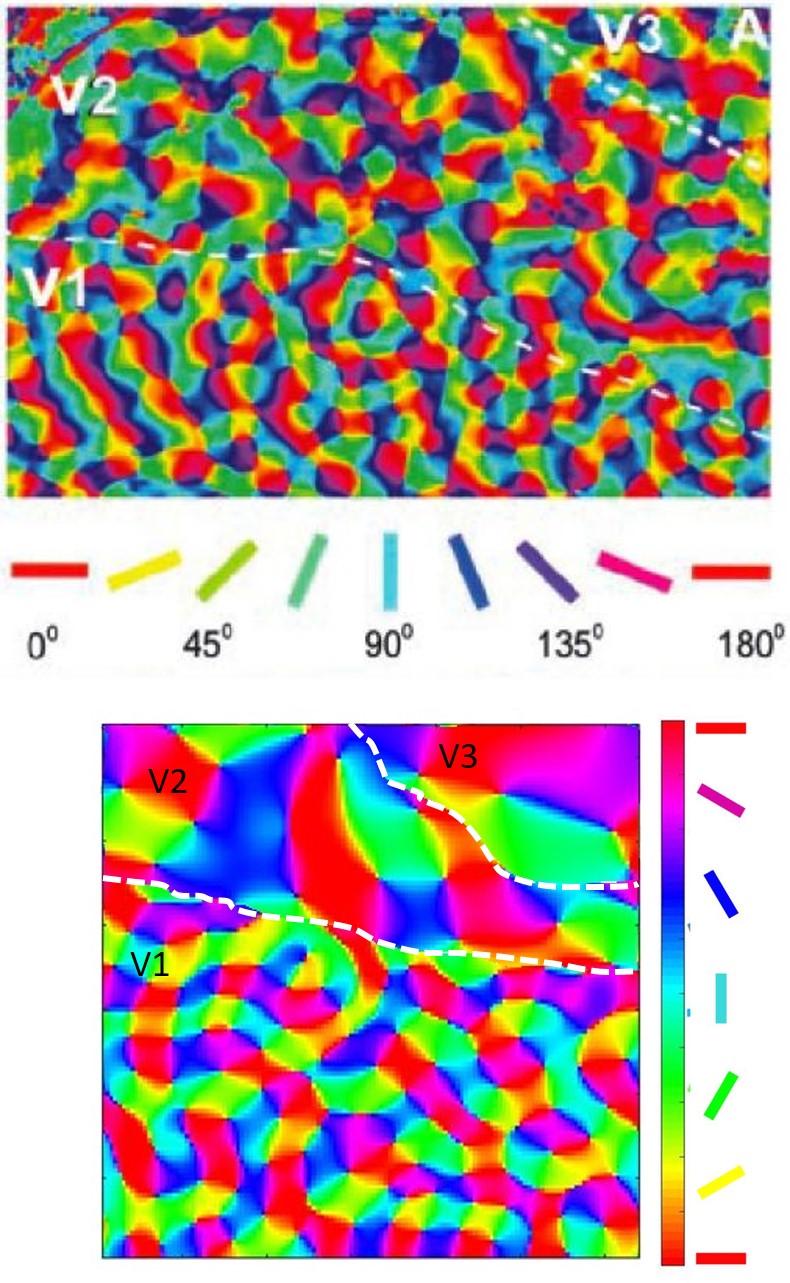}}
\caption{Top: The original neurophysiological results taken from \cite[Figure 37]{petitot2008neurogeometrie}. As one passes through V1-V2-V3 the size of the simple cells increase and the lattice of the orientation map extends while the pinwheels are preserved, Bottom: The simulation results obtained via our model. The model is able to produce the same type of orientation maps, which preserve the pinwheels through different cortex layers, by simply changing the scale of the Gabor filter bank.}
\label{fig:petitotComparisonFigure}
\end{figure}

Let us note that as one passes through V1-V2-V3 areas of the cortex the sizes of the simple cells increase and the lattice of the orientation map extends while the pinwheels are preserved, as visualized on the top of Figure \ref{fig:petitotComparisonFigure}. Our simulations are able to reproduce the same type of orientation maps, which preserve the pinwheels through different cortex layers, by simply changing the scale of the Gabor filter bank as shown on the bottom of Figure \ref{fig:petitotComparisonFigure}.

In the next series of experiments we will compute the orientation maps by selecting at every point orientation and scale by using the three equations \eqref{eq:filteringExpressionFin}, \eqref{eq:ourThetaCondition} and \eqref{eq:ourSigmaCondition}. 
This case is the closer one to the physiological situation of a normal visual cortex, where cells with different orientations and sizes are present.
In Figure \ref{fig:multiScale_FullProcedure2} the relevant simulation result of the model is visualized, showing the orientation map rendering both orientation and scale selectivity.

In the final experiment, which is given in Figure \ref{fig:petitotMultiScale}, we used the same procedure as in Figure \ref{fig:multiScale_FullProcedure2} but using three different sets of scales and we obtained a result similar to Figure \ref{fig:petitotComparisonFigure}. This procedure is closer to the real receptive field composition of the primary visual cortex. 

\begin{figure}
\centerline{\includegraphics[scale=0.3,trim={0cm 0cm 0cm 0cm},clip]{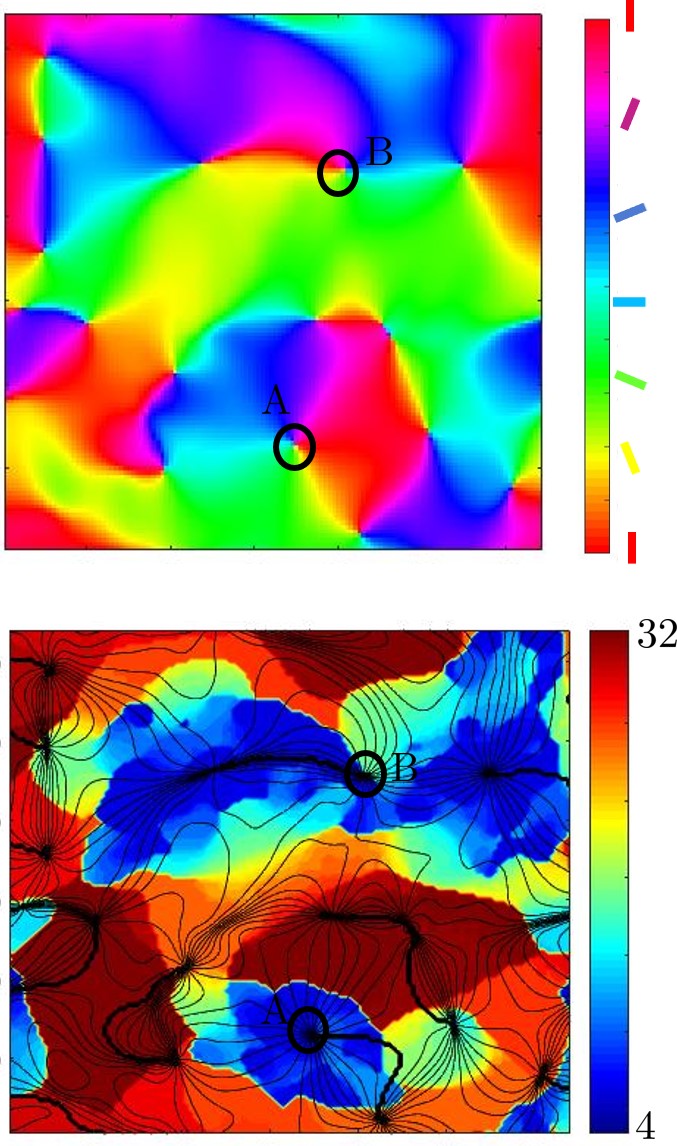}}
\caption{Top: Orientation preference map obtained through our model based on the procedure based on \eqref{eq:filteringExpressionFin}, \eqref{eq:ourThetaCondition} and \eqref{eq:ourSigmaCondition} with scale set $\{4,4.5,5,\dots 32\}$ on which maximum selectivity over the scale set is applied. Bottom: The corresponding scale map where each color indicates a certain scale value found by \eqref{eq:ourSigmaCondition} and black curves represent the iso-orientation lines.}
\label{fig:multiScale_FullProcedure2}
\end{figure}

\begin{figure}
\centerline{\includegraphics[scale=0.4,trim={0cm 0cm 0cm 0cm},clip]{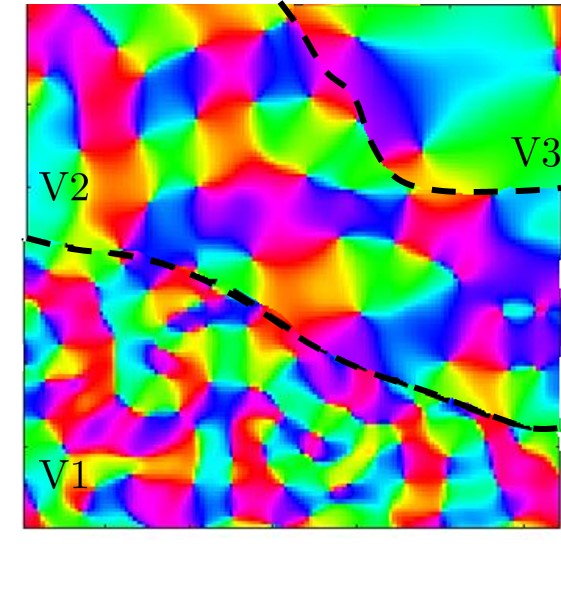}}
\caption{Orientation preference maps obtained through our model using the procedure based on \eqref{eq:filteringExpressionFin}, \eqref{eq:ourThetaCondition} and \eqref{eq:ourSigmaCondition} with scale sets $\{4,4.5,5,\dots 8\}$ for V1, $\{4,4.5,5,\dots 16\}$ for V2 and $\{4,4.5,5,\dots 32\}$ for V3.}
\label{fig:petitotMultiScale}
\end{figure}

\section{Conclusion}

In this paper we presented a new model for the generation of orientation preference maps in the primary visual cortex, considering both orientation and scale features.
We considered modeling the functional architecture of the primary visual cortex by taking into account orientation and scale features and using a framework inspired by Sarti et. al. \cite{sarti2008symplectic}. Furthermore, we also provided the physical reasoning behind the choice of the generalized Gabor function by showing that it is a coherent state of the non-commutative framework corresponding to the cortex functional architecture.
The intrinsic variables of orientation and scale constitute a fiber on each point of the retinal plane and the set of receptive profiles of simple cells is located on the fiber.  Orientation preference maps are then obtained simply as the lifting of a noise stimulus by a set of Gabor filters, mapping the orientation value on the 2-dimensional plane. This corresponds to a Bargmann transform in the reducible representation of the $\se$ group which is followed by a maximum response selection procedure. A comparison has been provided with a previous model based on the Bargmann transform in the irreducible representation of the $\se$ group, outlining that the new model is more physiologically motivated.
From simulation results appears that this technique is able to reproduce cortical maps of different areas with morphological characteristics comparable to experimental data.
A clear advantage of the method consists also in its versatility since a number of different features could be considered, such as frequency and phase. Further studies will be conducted in this direction in the next future.

\FloatBarrier

\bibliographystyle{spmpsci}
\bibliography{JMIV_Bib}

\end{document}